\documentclass[11pt]{article}
\usepackage{amsmath, amssymb, amsthm, xparse, xcolor}
\usepackage[normalem]{ulem}
\usepackage[margin=1in]{geometry}
\usepackage{dsfont}
\usepackage[parfill]{parskip}
\usepackage{url}
\usepackage{mathrsfs}
\usepackage{tcolorbox}
\tcbuselibrary{breakable}
\usepackage{caption}
\usepackage{enumitem}
\usepackage{thm-restate}
\usepackage{microtype}
\usepackage{subcaption}
\usepackage{setspace}

\usepackage{tikz}
\usetikzlibrary{positioning}
\usetikzlibrary{shapes.geometric, arrows}
\usetikzlibrary{decorations.pathreplacing}

\usepackage{float}
\usepackage{sn-preamble} 
\usepackage{wasysym}
\SetSymbolFont{wasy}{bold}{U}{wasy}{m}{n}

\NewDocumentCommand{\ot}{g}{
    \IfNoValueTF{#1}{\tilde{O}}{\tilde{O}(#1)}
}

\newcommand{\Unif}{\mathrm{Unif}}

\newcommand{\err}{\operatorname{\mathsf{err}}}

\newtheorem{abort}[theorem]{Stopping Condition}

\newcommand{\widebar}[1]{%
  \mkern2.5mu
  \ooalign{%
    $\overline{\mkern-2.5mu#1\mkern-2.5mu}$\cr
    \hidewidth
    \raisebox{0.07ex}{$\overline{\mkern-2.5mu\phantom{#1}\mkern-2.5mu}$}%
    \hidewidth\cr
  }%
  \mkern2.5mu
}

\title{ Quasi-Monte Carlo Beyond Hardy--Krause II:\\ $(1 + \varepsilon) n$ Samples Suffice  \vspace{10pt} }

\author{
Ekene Ezeunala\thanks{University of Chicago, Chicago, IL, USA. \texttt{ekene@uchicago.edu.}}
\and 
Agastya Vibhuti Jha\thanks{University of Chicago, Chicago, IL, USA. \texttt{agastyajha@uchicago.edu}.}
\and
Haotian Jiang\thanks{University of Chicago, Chicago, IL, USA. \texttt{jhtdavid@uchicago.edu}.}}

\date{}

\begin{document}

\pagenumbering{gobble}

\maketitle

\begin{abstract}
Numerical integration studies how well one can estimate the integral of a function $f$ over $[0,1)^d$ using $n$ sample points. The two classical methods, Monte Carlo (MC) and quasi-Monte Carlo (QMC), have complementary strengths and weaknesses, and a fundamental question is to design an approach that combines the benefits of both. 
Recently, building on the transference principle in discrepancy theory, Bansal and Jiang~\cite{BJ25a} gave a randomized QMC method that bridges MC and QMC guarantees using only i.i.d.\ samples. Their method also goes beyond the classical Koksma--Hlawka inequality: it achieves integration error $\widetilde{O}_d(\sigma_{\mathsf{SO}}(f)/n)$, where the smoothed-out variation $\sigma_{\mathsf{SO}}(f)$ can be substantially smaller than the Hardy--Krause variation that governs the classical bound. However, their algorithm requires $n^2$ i.i.d.\ samples as input, and this quadratic blowup is inherent to any method based on the transference principle.

\smallskip

In this work, we bypass the quadratic blowup: for any constant $\varepsilon > 0$, we show that $(1+\varepsilon)n$ i.i.d.\ samples suffice to both obtain the beyond-Hardy--Krause guarantee of~\cite{BJ25a}, resolving an open problem posed there, and to produce low-discrepancy point sequences. Our algorithms are variants of the online Haar-thinning method of Dwivedi, Feldheim, Gurel-Gurevich, and Ramdas~\cite{DFG+19}. Specifically, from $(1+\varepsilon)n$ i.i.d.\ samples:
\medskip
\begin{enumerate}
  \item The \textsl{uniformly-shifted Haar-thinning} algorithm retains $n$ samples with integration error $\widetilde{O}_d(\sigma_{\mathsf{SO}}(f)/n)$, matching the bound of Bansal and Jiang \cite{BJ25a}.

  \smallskip
  
  \item The \textsl{linear-feedback Haar-thinning} algorithm retains $n$ samples with star discrepancy $O_d(\log^{d+1} n)$, improving the $O_d(\log^{2d+1} n)$ bound of~\cite{DFG+19} and attaining their conjectured bound. 
  Our bound nearly matches the best known $O_d(\log^d n)$ bound achieved by explicit constructions of point sequences \cite{Hal60,Sob67,Fau82,Nie87}. 
\end{enumerate}

\medskip

Both algorithms are online, never discard two consecutive samples, and spend $\widetilde{O}_d(1)$ time to process each sample. 
At the heart of our analysis is a new property of Haar-thinning strategies, which we establish via a symmetry argument: the discrepancies  of distinct low-order Haar functions are uncorrelated.

\end{abstract}

\newpage

\newpage
\setcounter{tocdepth}{2}

\begingroup
\hypersetup{linkcolor=BrickRed}
\begin{spacing}{1.3}
\small
\tableofcontents
\end{spacing}
\endgroup
\newpage

\setcounter{page}{1}
\pagenumbering{arabic}

\section{Introduction}\label{sec:intro}

Numerical integration estimates the integral $\smash{\bar{f} = \int_{[0,1)^d} f(\bz)\,\mathrm d\bz}$ by averaging $f$ over an $n$-point set $A \subset [0,1)^{d}$, i.e., $\smash{\bar{f}(A) = |A|^{-1}\sum_{\bz \in A} f(\bz)}$. It arises in many applications, e.g., in finance, graphics, statistics, and scientific computing \cite{Gla03,Lem09,Owe13}, where $\bar{f}$ cannot be computed exactly because $f$ is complicated or accessible only through pointwise evaluations. The main goal of numerical integration is to choose $A$ to minimize the magnitude of the error $\smash{\err(A,f) = \bar{f}(A) - \bar{f}}$. Broadly speaking, there are two main paradigms for choosing $A$. 

\smallskip
\noindent \textbf{Monte Carlo Methods.} The Monte Carlo (MC) method chooses the $n$ points of $A$ independently and uniformly at random. 
As it requires only random samples, MC is highly flexible and widely used in practice. 
The downside is its slow convergence rate: its mean-squared error is
\begin{align}\label{eq:MCError}
\mathbb{E}\bigl[\err(A,f)^2\bigr] = \sigma^{2}(f)/n,
\end{align}
where $\smash{\sigma(f) := \bigl(\E_{\bz\sim[0,1]^d}(f(\bz)-\bar{f})^2\bigr)^{1/2}}$ is the standard deviation of $f$. 
Hence the error decays only as $n^{-1/2}$. Equivalently, $\sigma^2(f)/\delta^2$ samples are needed to achieve error $\delta$, which can be prohibitive when samples are expensive to obtain or $f$ is expensive to evaluate (e.g., \cite{Gla03,Fis05}).

\smallskip
\noindent \textbf{Quasi-Monte Carlo Methods.} Quasi-Monte Carlo (QMC) methods give up the flexibility of random sampling and instead carefully choose a deterministic point set $A$. The classical Koksma--Hlawka inequality \cite{Kok42,Hla61,Zar68} bounds the QMC error as
\begin{align}\label{eq:StdQMCError}
    |\err(A,f)| \leq \frac{1}{n} \, V_{\mathsf{HK}}(f)\cdot D^{\star}(A) .
\end{align}
Here, $V_{\mathsf{HK}}(f)$ is the Hardy--Krause variation of $f$, and $D^*(A) = \sup_{R} \big||A\cap R| - n |R|\big|$ is the {\em star} discrepancy  of $A$, where the supremum is over all anchored boxes $R$ of the form $[0,\bz_1) \times [0, \bz_2) \times \cdots \times [0,\bz_d)$. 
Intuitively, $D^*(A)$ measures the uniformity of $A$ with respect to anchored boxes.\footnote{In the literature, the star discrepancy is also often referred to as the continuous discrepancy or the Lebesgue measure discrepancy for anchored boxes. Anchored boxes are also often referred to as corners.}

As the Koksma--Hlawka inequality \eqref{eq:StdQMCError} is tight in general, much of the work on QMC has focused on constructing point sets with small $D^{\star}(A)$ (e.g., see \cite{Nie92,DP10}).
The best known point sequences\footnote{A \emph{point sequence} is an infinite sequence $\bz_1, \bz_2, \ldots \in [0,1)^d$ whose first $n$ points are required to have small star discrepancy simultaneously for every $n$. A \emph{point set}, in contrast, is designed for a single fixed $n$ and the construction may differ for different $n$. The point set notion is strictly weaker: the best known $n$-point sets achieve $D^{\star} = O_d(\log^{d-1} n)$ \cite{Ham60,Nie87,DP10}, while the best known sequences achieve $O_d(\log^{d} n)$ \cite{Hal60,Fau82,Nie87}; for $d=1$ this $\log n$ gap is provably necessary \cite{Sch72}. 
We work with sequences: they are the standard object in QMC practice, where $n$ is rarely known in advance, and they are what our online algorithms naturally produce.} achieve $D^{\star}(A) = O_d(\log^d n)$\footnote{Throughout, $O_d(\cdot)$ hides $\exp(O(d))$ factors, and $\widetilde{O}_d(\cdot)$ hides $O(\log n)^{O(d)}$ factors. In the context of QMC, the dimension $d$ should be viewed as a constant.}  for all $n$ simultaneously, where $A$ is the set of the first $n$ points of the sequence \cite{Hal60,Sob67,Fau82,Nie87}. Plugging this bound into \eqref{eq:StdQMCError} gives an
error of $\widetilde{O}_d(V_{\mathsf{HK}}(f)/n)$, i.e., a faster convergence rate of $\widetilde{O}_d(1/n)$ in place of the $n^{-1/2}$ rate of MC, which is what makes QMC methods widely applied in practice~\cite{Lem09,Mat09,DKPS13,Owe13}.

Despite their faster convergence rate, QMC methods have two well-known  limitations. First, for fast-varying functions, $V_{\mathsf{HK}}(f)$ can be much larger than $\sigma(f)$ (e.g.,
$V_{\mathsf{HK}}(f) = \Theta(k)$ but $\sigma(f) = \Theta(1)$ for $f(x) = \sin(kx)$), so the QMC bound \eqref{eq:StdQMCError} can be far worse than the MC
bound \eqref{eq:MCError}. Second, QMC evaluates $f$ on a fixed point set, which is  impossible in applications where only random samples are available.
Given these complementary strengths and weaknesses, a fundamental problem in numerical integration is to design a method that combines the benefits of both MC and QMC. 

\smallskip
\noindent \textbf{Quasi-Monte Carlo Beyond Hardy--Krause.}
Recently, Bansal and Jiang \cite{BJ25a} made substantial progress in this direction, bypassing previous limitations \cite{Mat09,Lem09,DP10,Owe13}.
Specifically, they gave a randomized QMC method that achieves not only the best of MC and QMC guarantees, but also a surprising improvement over Koksma--Hlawka inequalities: their algorithm 
has error $\smash{\widetilde{O}_d(\sigma_{\mathsf{SO}}(f)/n)}$, where $\sigma_{\mathsf{SO}}(f)$ is the \emph{smoothed-out variation} of $f$, which is substantially smaller than the Hardy--Krause variation $V_{\mathsf{HK}}(f)$ in \eqref{eq:StdQMCError}. 
Their main result is the following. 

\begin{restatable}[\cite{BJ25a}]{theorem}{BansalJiangTheorem}
\label{thm:BJ25}
    There is an online randomized algorithm that takes as input $n^2$ i.i.d. uniform samples from $[0,1)^d$ and, in  $\widetilde{O}_d(n^2)$ time, partitions them into $n$ sets each of size $n$. Each of these sets $A$ satisfies the following: for any function $f\in L^2([0,1)^{d})$ with finite $\sigma_{\mathsf{SO}}(f)$,
\[
\mathbb{E}\bigl[\mathsf{err}(A,f)^2\bigr]
\leq
\widetilde{O}_{d}\big(\sigma_{\mathsf{SO}}^2(f) / n^2 \big) .
\]
\end{restatable}

\smallskip
\noindent \textbf{The Need for $n^{2}$ Samples.} 
Despite its strong guarantees, the algorithm in \cite{BJ25a} requires $n^2$ samples to begin with and $\widetilde{O}_d(n^2)$ running time, even to produce a single QMC set. 
This quadratic blowup in the number of initial samples (and hence running time) comes from their use of the \emph{transference principle} in geometric discrepancy theory (e.g., see \cite{Mat09,ABN16}), and is inherent to any method based on this principle. 
To address this issue, Bansal and Jiang asked \cite[Section 6]{BJ25a} whether there is an algorithm that can achieve their improvement over Hardy--Krause variation and the Koksma--Hlawka inequality but only requires $\widetilde{O}_d(n)$ samples and $\widetilde{O}_d(n)$ running time.

\subsection{Our Results}
\label{subsec:results}

Our main result is a positive answer to the above question of Bansal and Jiang. In fact, we show that $(1+\varepsilon)n$ samples already suffice to achieve their beyond-Hardy--Krause error bound.\footnote{Throughout, the dependence on $\varepsilon$ in our bounds is $\poly(1/\varepsilon)$. The polylogarithmic factor hidden in $\widetilde{O}_{d,\varepsilon}(\cdot)$ in \Cref{thm:main} is $\log^{d} n$, which also matches that of \cite{BJ25a}.}

\begin{restatable}[QMC beyond Hardy--Krause via $(1+\varepsilon)n$ samples]{theorem}{MainThm}
\label{thm:main}
Fix an arbitrary constant $\varepsilon \in (0,1)$.
There is a randomized algorithm that, for any $n \in \mathbb{Z}_{>0}$, takes as input $(1+\varepsilon)n$ i.i.d.\ uniform samples from $[0,1)^d$ and outputs with high probability a subset of $n$ samples $A$ such that for any function $f \in L^2([0,1)^{d})$ with finite $\sigma_{\mathsf{SO}}(f)$,
\[
\mathbb{E}\bigl[\mathsf{err}(A,f)^2\bigr]
\leq
\widetilde{O}_{d,\varepsilon}\big(\sigma_{\mathsf{SO}}^2(f) / n^2 \big).
\]
Moreover, the algorithm is online, never discards two consecutive samples, and spends $\widetilde{O}_d(1)$ time to process each sample.
\end{restatable}

\smallskip
\noindent \textbf{Online Thinning.}
To bypass the quadratic blowup in the number of initial samples, we abandon the transference principle entirely and instead consider the \emph{online thinning} framework of Dwivedi, Feldheim, Gurel-Gurevich, and Ramdas \cite{DFG+19}. Here, samples arrive one at a time, and upon seeing a sample, the algorithm must decide whether to retain or discard it (based on the samples retained so far), and in case it discards a sample, it must retain the next one.

In this model, \cite{DFG+19} gave an elegant algorithm called {\em Haar-thinning}: from $(1+\varepsilon)n$ samples, with high probability, it retains $n$ samples with star discrepancy $O_{d,\varepsilon}(\log^{2d+1} n)$, which already recovers the classical QMC bound of $\widetilde{O}_{d,\varepsilon}(V_{\mathsf{HK}}(f)/n)$ by the Koksma--Hlawka inequality \eqref{eq:StdQMCError}. 

To achieve the beyond-Hardy--Krause guarantee in \Cref{thm:main} and \cite{BJ25a}, we use a {\em uniformly-shifted} variant of the Haar-thinning strategy. Key to our analysis is a new property of Haar-thinning strategies, which we establish via a novel symmetry argument: the discrepancies of distinct low-order Haar functions are uncorrelated (see \Cref{subsec:uncorrelation-Haar}).

\smallskip
\noindent \textbf{Low-Discrepancy Sequences via Online Thinning.}
Our second contribution is an improvement to the $O_{d,\varepsilon}(\log^{2d+1} n)$ star discrepancy bound of \cite{DFG+19}. Dwivedi,
Feldheim, Gurel-Gurevich, and Ramdas conjectured in \cite[Conjecture 2]{DFG+19} that a greedy version of their Haar-thinning strategy attains star discrepancy $O_{d,\varepsilon}(\log^{d+1} n)$. 
While it is still unclear how to analyze their proposed greedy strategy, we give a new variant of Haar-thinning that we call {\em linear-feedback Haar-thinning} which attains their conjectured $O_{d,\varepsilon}(\log^{d+1} n)$ bound.

\begin{restatable}[Low-discrepancy point sequences via $(1+\varepsilon)n$ samples]{theorem}{gridRectangleThinning}
\label{thm:GridRectangleThinning}
Fix an arbitrary constant $\varepsilon\in(0,1)$. There is a randomized algorithm which, given $n \in \mathbb{Z}_{>0}$, takes as input $(1+\varepsilon)n$ i.i.d.\ uniform samples from $[0,1)^d$ and outputs with high probability a subset of $n$ samples $A$ that satisfies 
\[
D^{\star}(A) \leq O_{d,\varepsilon}(\log^{d+1}n).
\]
Moreover, the algorithm is online, never discards two consecutive samples, and spends $\widetilde{O}_{d,\varepsilon}(1)$ time to process each sample.
\end{restatable}

The algorithm in \Cref{thm:GridRectangleThinning} produces a point sequence, matching the best known $O_d(\log^d n)$ bound for explicit point sequences \cite{Hal60,Sob67,Fau82,Nie87} up to a $\log n$ factor.

\smallskip
\noindent \textbf{Roadmap.} We begin with an overview of our techniques in  \Cref{subsec:overview}. Then in \Cref{sec:preliminaries}, we collect  the preliminaries and describe the online thinning framework from \cite{DFG+19}. We present the uniformly-shifted Haar-thinning algorithm in \Cref{sec:qmc} and prove \Cref{thm:main}. Finally, we give the linear-feedback Haar-thinning algorithm for \Cref{thm:GridRectangleThinning} in  \Cref{sec:rect-disc}.

\subsection{Overview}\label{subsec:overview}
In this subsection, we give an overview of our techniques. 
We start with a brief recap of Bansal and Jiang's QMC method \cite{BJ25a} and explain why it requires $n^2$ initial samples. We then discuss the Haar-thinning strategy of \cite{DFG+19}, and present our variants of it that attain the QMC bound in \cite{BJ25a} and an improved $O_{d,\varepsilon}(\log^{d+1} n)$ star discrepancy bound using only $(1+\varepsilon)n$ samples.

\subsubsection{The Bansal--Jiang QMC Method}
\label{subsubsec:Bansal-Jiang}

The Bansal--Jiang method is based on the
transference principle. It starts from $n^2$ i.i.d. uniform samples from $[0,1)^d$, and iteratively thins them down as follows: at each step, it computes a balanced $\{\pm 1\}$-coloring of the current samples with both low discrepancy and 
sufficient randomness for dyadic rectangles (of side length $\geq 1/\poly(n)$)---the two colors then split the current samples into two equal-sized halves. Iterating this process for $\log n$ steps splits the $n^2$ samples into $n$ sets of size $n$ each. 
Using the properties above, Bansal and Jiang found cancellations missed by the Koksma--Hlawka inequality \eqref{eq:StdQMCError}, which led to their beyond-Hardy--Krause error bound in \Cref{thm:BJ25}. 

\smallskip
\noindent \textbf{The Inherent Quadratic Blowup.} Despite the strong guarantees, Bansal and Jiang's method requires $n^2$ samples to begin with, even to produce a single QMC set. This quadratic blowup is  in fact inherent to any method based on the transference principle: the integration error from any such method is at least the Monte Carlo error of the initial sample pool, which can only be bounded by $O(1/n)$ when the initial sample size exceeds $\Omega(n^2)$.

\subsubsection{The Haar-Thinning Strategy}
\label{subsubsec:online}

To bypass the quadratic blowup, we completely abandon the transference principle, and consider the alternative online thinning framework of \cite{DFG+19}.
Here, samples arrive one at a time, and upon seeing a sample, the algorithm must decide whether to retain or discard it based on the samples retained so far; and in case it discards a sample, it must retain the next one.
Clearly, the amount of over-sampling is determined by the probability of rejecting each sample. 

In this model, Dwivedi,
Feldheim, Gurel-Gurevich, and Ramdas
\cite{DFG+19} gave a beautiful strategy called {\em Haar-thinning} that, from $(1+\varepsilon)n$ samples, retains $n$ samples (by rejecting each sample with probability at most $\varepsilon$) with star discrepancy $O_d(\log^{2d+1} n)$. This
already recovers the classical QMC error bound of $\widetilde{O}_d(V_{\mathsf{HK}}(f)/n)$ by the Koksma--Hlawka inequality \eqref{eq:StdQMCError}. 

\smallskip
\noindent \textbf{Haar-Thinning as Rejection Sampling.} Roughly speaking, the Haar-thinning strategy is a rejection sampling method based on ``votes'' from all Haar functions (up to order $\ell :=O(\log n)$): 
upon receiving each sample $\bx_t$, the Haar functions whose discrepancies will increase (resp. decrease) from retaining $\bx_t$ will vote ``REJECT'' (resp. ``ACCEPT''), and those not affected by $\bx_t$ will not vote. 
Here, the discrepancy\footnote{Intuitively, as each Haar function is mean-zero (i.e., $\overline{H} = 0$), $\varphi_t$ can be viewed as a continuous notion of discrepancy for functions which is analogous to the star discrepancy $D^*(A_t)$.} of each Haar function $H$ at step $t$ is $\varphi_t(H) := \sum_{\bz \in A_t} H(\bz)$, where $A_t$ denotes the samples retained prior to step $t$. 
Then the Haar-thinning strategy rejects the sample $\bx_t$ (and hence commits to the next sample) with probability roughly 
\begin{align} \label{eq:Haar-thinning-density-overview}
\frac{\varepsilon}{2} \cdot \Big(1 + \frac{\#\textsf{REJECT} - \#\textsf{ACCEPT}}{\ell^d} \Big) .
\end{align}
Note that since any sample $\bx_t$ lies in at most $\ell^d$ Haar functions, the rejection probability in \eqref{eq:Haar-thinning-density-overview} is always in $[0,\varepsilon]$, so this strategy indeed uses at most $(1+\varepsilon)n$ samples.

An ingenious insight used in \cite{DFG+19} to prove their star discrepancy bound is the following: because distinct Haar functions are orthogonal, the rejection rule \eqref{eq:Haar-thinning-density-overview} actually reduces the discrepancy of {\em every} Haar function on average. In particular, the slight bias in \eqref{eq:Haar-thinning-density-overview} ensures that each $\varphi_t(H)$ has a negative drift of $O_\varepsilon(1/\ell^d)$ towards $0$, and hence cannot exceed $O_\varepsilon(\ell^{d+1}) = O_{\varepsilon}(\log^{d+1} n)$ with high probability.   Their $O_{d,\varepsilon}(\log^{2d+1} n)$ star discrepancy bound then follows from a standard Haar decomposition of each anchored box, which involves at most $O_d(\log^d n)$ Haar functions.

\subsubsection{Beyond-Hardy--Krause via Uncorrelation of Haar Discrepancies}

To recover the beyond-Hardy--Krause guarantee in \cite{BJ25a}, our key technical insight is a new {\em uncorrelation} property for the Haar-thinning strategy: at any step $t$, while the discrepancies of different Haar functions are highly dependent due to their overlaps, they are in fact uncorrelated, i.e., $\E[\varphi_t(H) \varphi_t(G)] = 0$ for any distinct Haar functions $H$ and $G$ (up to order $\ell$).

This property cannot be established by a per-step analysis, as the bias in the voting rule \eqref{eq:Haar-thinning-density-overview} (and hence the density of the next retained sample) breaks the orthogonality of the Haar basis. Instead, we use a symmetry argument---by defining a bijection $T$ that flips the two halves of $H$ in some dimension (so that $\varphi_t(H)$ flips its sign while $\varphi_t(G)$ remains the same), we show that the Haar-thinning strategy commutes with the mapping $T$ and thus $\varphi_t(H) \varphi_t(G)$ and $-\varphi_t(H) \varphi_t(G)$ have the same distribution. We postpone the details to \Cref{subsec:uncorrelation-Haar}.   

\smallskip
\noindent \textbf{Uniformly-Shifted Haar-Thinning.} With the above uncorrelation property, we consider a uniformly-shifted variant of the Haar-thinning strategy where the Haar basis functions are all shifted by $\bs\sim \mathsf{Unif}([0,1)^d)$. For the integration error, one can equivalently view the shift $\bs$ as being applied to the function $f$, while the algorithm remains the same as the Haar-thinning strategy. 

We show that this algorithm achieves the beyond-Hardy--Krause error bound in \Cref{thm:main} via the streamlined analysis in \cite{CJJ26} that directly works with the Haar basis. In particular, the uncorrelation property for the Haar discrepancies gives the cancellation needed to go beyond the Hardy--Krause variation; and the random shift $\bs$ allows one to pass from the integration error in the Haar basis, which is captured by the Haar--Besov seminorm (see \Cref{defn:HaarBesovSeminorm}), to the smoothed-out variation (see \Cref{thm:shift-averaged-mixed-haar-equivalence}). We leave the details to \Cref{sec:qmc}.

\subsubsection{Achieving Lower Star Discrepancy via Online Thinning}

The uncorrelation of Haar discrepancies turns out to buy us more than just the beyond-Hardy--Krause error bound in \cite{BJ25a} for numerical integration: it also allows us to improve upon the star discrepancy bound of \cite{DFG+19} as follows.

Specifically, in their analysis (recall the last paragraph of \Cref{subsubsec:online}), when writing each anchored box as the weighed sum of $O_d(\log^d n)$ Haar functions, the uncorrelation property allows these Haar discrepancies to add up in an $\ell_2$ way---this directly gives an improved $O_{d,\varepsilon}(\log^{3d/2 + 1} n)$ star discrepancy bound for the Haar-thinning strategy, settling \cite[Conjecture 1]{DFG+19}.

However, the Haar-thinning strategy falls short of achieving the $O_d(\log^{d+1} n)$ bound in \Cref{thm:GridRectangleThinning}. The main bottleneck here is that the drift towards $0$ is only on the order of $O_\varepsilon(1/\log^d n)$ for each Haar function, and hence the discrepancy of each single Haar function may go up to $O_\varepsilon(\log^d n)$. 

\smallskip
\noindent \textbf{Linear-Feedback Haar-Thinning.} To bypass the above issue, we consider a linear-feedback variant of the Haar-thinning strategy: instead of just letting each Haar function vote ``REJECT'' or ``ACCEPT'', we also have each vote carry the weight of the corresponding Haar discrepancy and use a rejection sampling rule based on the weighted difference of the votes:
\begin{align} \label{eq:linear-feedback-density-overview}
\frac{\varepsilon}{2} \cdot \Big(1 + \frac{\sum_{H: \textsf{ REJECT}} |\varphi(H)| - \sum_{H: \textsf{ ACCEPT}} |\varphi(H)|}{\ell^d} \Big) .
\end{align}
Intuitively, the Haar functions with larger discrepancies now have a stronger drift towards $0$ that is proportional to their discrepancies. This allows us to control the Haar discrepancies to be on the much smaller order of $O_{d,\varepsilon}(\log^{d/2+1} n)$, which leads to \Cref{thm:GridRectangleThinning}.   

Note that one potential issue with the rejection sampling rule \eqref{eq:Haar-thinning-density-overview} is that the total weights of all Haar functions affected by the current sample $\bx_t$ may go up to $(\log^{d/2+1} n) \cdot \ell^d$, and hence \eqref{eq:Haar-thinning-density-overview} might be an invalid rejection probability. 
Nonetheless, we are able to show that \eqref{eq:Haar-thinning-density-overview} remains bounded in $[0,\varepsilon]$ with high probability via a delicate martingale analysis. See \Cref{sec:rect-disc} for the details.

\subsection{Further Related Work}

\noindent \textbf{Low-Discrepancy Point Sets and Sequences.}
Low-discrepancy point sets and sequences  have been central to geometric discrepancy theory since the ground-breaking work of van der Corput \cite{vdC35a,vdC35b}, who showed that the star discrepancy for $d=2$ is $O(\log n)$ as opposed to $O(\sqrt{n})$ attained by i.i.d. random samples or the $\sqrt{n} \times \sqrt{n}$ grid. 
The result of van der Corput has since been extended, giving the current best star discrepancy bound of $O_d(\log^{d-1} n)$ for point sets (e.g., \cite{Ham60,Nie87,Nie92}) and $O_d(\log^d n)$ for point sequences (e.g., \cite{Hal60,Sob67,Fau82,Nie87}). 
These constructions have found important applications to numerical integration \cite{Kok42,Hla61,Zar68} and have been extensively studied in this context since then \cite{Lem09,DP10,DKPS13,Owe13}.

Interestingly, the optimal star discrepancy bound for point sets and sequences remains unknown when $d \geq 3$. For $d = 2$, Schmidt \cite{Sch72} proved the optimal lower bound of $\Omega(\log n)$ for point sets (or point sequences in 1-d). Later, Hal\'asz \cite{Hal81} gave an elegant new proof of Schmidt's lower bound using a Riesz product. For $d\geq 3$, the best lower bound is essentially $\Omega_d(\log^{(d-1)/2} n)$ from Roth's seminal work \cite{Rot54}, with small improvements given in \cite{Bec89,BL08,BLV08}.

\noindent \textbf{Non-Uniform Distribution.} While the uniform distribution is standard when studying the star discrepancy or numerical integration, the generalizations of these notions to non-uniform distributions have also been studied. We refer readers to \cite{AD14,ABN16} and the references therein. 

Variants of the online thinning problem where the underlying distribution is non-uniform have also been studied recently. In these contexts, the Haar functions seem less useful as they no longer form an orthogonal basis. We refer readers to the recent work \cite{SV26a,SV26b} for details.

\section{Preliminaries}\label{sec:preliminaries}

This section introduces the notation and preliminaries used throughout the paper.

\noindent \textbf{Haar Functions and Boxes.} For integers $j\geq 1$ and $0\leq k<2^{j-1}$, let $h_{j,k}$ denote the one-dimensional Haar function that equals $1$ on $[(2k)2^{-j},(2k+1)2^{-j})$, equals $-1$ on $[(2k+1)2^{-j},(2k+2)2^{-j})$, and is zero elsewhere. We also let $h_{0,0}$ denote the constant function $1$ on $[0,1)$. For vectors $\mathbf{j}$ and $\mathbf{k}$, define the multidimensional Haar function $H_{\mathbf{j},\mathbf{k}}$ on $[0,1)^d$ by $\smash{H_{\mathbf{j},\mathbf{k}}(\bx)=\prod_{i\in[d]}h_{\mathbf{j}_i,\mathbf{k}_i}(\bx_i)}$. We write $|\mathbf{j}|$ for the $\ell_1$-norm of $\mathbf{j}$, namely $\sum_{i\in[d]}\mathbf{j}_i$. We note that the Haar functions form an orthogonal basis for $\smash{L^2([0,1)^d)}$ under the inner product $\smash{\langle f,g\rangle=\int_{[0,1)^d}f(\bx)g(\bx)\,\mathrm d\bx}$.

We abuse notation and use $\Pi_{\leq\ell}$ to denote both the collection of non-constant Haar functions $H_{\mathbf{j},\mathbf{k}}$ with $\mathbf{j}_i\leq\ell$ for every $i\in[d]$ and the orthogonal projection onto their span. Likewise, $\Pi_{>\ell}$ denotes both the remaining nonconstant Haar functions and the orthogonal projection onto their span. Further, we use $N_{\ell}$ to denote the $|\Pi_{\leq \ell}|$.

\smallskip
\noindent \textbf{Haar--Besov Seminorm and Equivalence.} The Haar--Besov seminorm is central to our analysis. We refer the reader to \cite{vybiral2006function,triebel2019function} for background.
\begin{definition}[Haar--Besov Seminorm]\label{defn:HaarBesovSeminorm}
The Haar--Besov seminorm of $f\in L^2([0,1)^d)$ is
\[
\lVert f\rVert_B^2
:=
\sum_{\mathbf{j},\mathbf{k}}
\frac{
\bigl|\langle f,H_{\mathbf{j},\mathbf{k}}\rangle\bigr|^2}{\lVert H_{\mathbf{j},\mathbf{k}}\rVert_{2}^{4}},
\]
where $\mathbf{j},\mathbf{k} \in \mathbb{Z}^d$ are all tuples that satisfy $\mathbf{j}_i\geq 0$ and $0\leq\mathbf{k}_i<2^{\mathbf{j}_i}$ for every $i\in[d]$. 
\end{definition}

For a function $f$ on $[0,1)^d$ and a shift $\bs\in[0,1)^d$, define $f_{\bs}(\bx)=f(\bx+\bs\bmod[0,1)^d)$. 
It was shown in~\cite{CJJ26} that the  average of the squared Haar--Besov seminorm of $f_{\bs}$ over a uniform random shift $\bs$ is equivalent to the squared \emph{smoothed-out} variation $\sigma_{\mathsf{SO}}^2(f)$ introduced in~\cite{BJ25a}.

\begin{restatable}[Shift-averaged seminorm equivalence]
    {theorem}{shiftAveragedHaarEquivalence}
\label{thm:shift-averaged-mixed-haar-equivalence}
For any \(f\in L^2([0,1)^d)\), its average Haar--Besov seminorm over random shifts $\bs \sim \mathsf{Unif}([0,1)^{d})$ is equivalent to its smoothed-out variation:
\[
    \mathbb{E}_{\mathbf{s}\sim [0,1)^{d}}
    [\lVert f_{\mathbf{s}}\rVert_{B}^{2}]
    = 
    \Theta_{d}(\sigma_{\mathsf{SO}}(f)^2),
\]
where $\sigma_{\mathsf{SO}}(f)^2 = \sum_{\mathbf{k}\in\mathbb{Z}^d \setminus \{\mathbf{0}\}} |\widehat{f}_{\mathbf{k}}|^{2}\prod_{i \in [d]}\max(1, |\mathbf{k}_{i}|)$ and $\widehat{f}_{\mathbf{k}}$ denotes the Fourier coefficient corresponding to frequency $\mathbf{k}$ in $\mathbb{Z}^{d}$.
\end{restatable}
We do not discuss $\sigma_{\mathsf{SO}}$ in detail because our analysis works entirely in the Haar basis. We use \Cref{thm:shift-averaged-mixed-haar-equivalence} as a black box to translate our final bound into this error characterization.

\smallskip
\noindent \textbf{Boxes and Star Discrepancy.}
We use $\mathcal{R}$ to denote the set of all \emph{anchored boxes}, i.e., boxes of the form $[0,\bx_1) \times \cdots \times [0, \bx_d)$ for $\smash{\bx\in[0,1)^d}$.
We use $\mathcal{R}_{\ell}$ to denote the set of all {\em grid boxes}, where each edge is of the form $[u2^{-\ell},v2^{-\ell})$ for integers $0\leq u\leq v\leq 2^\ell$.  

For a finite point set $A\subset[0,1)^d$ and any Lebesgue-measurable set $R \subseteq [0,1)^d$, define the {\em continuous discrepancy}, or the {\em Lebesgue-measure discrepancy}, of $A$ with respect to $R$ as
\begin{equation}\label{eq:ContinuousDiscDefn}
    D_R(A) :=|A\cap R|-|A||R|,
\end{equation}
where $|R|$ is the Lebesgue measure of $R$. The {\em star discrepancy} of $A$ is defined as
\begin{equation}\label{eq:starDiscDefn}
    D^{\star}(A) :=\sup_{R \in \mathcal{R}}|D_{R}(A)|.
\end{equation}

\subsection{The \texorpdfstring{$(1 + \varepsilon)$}{1+eps}-Thinning Framework}\label{subsec:thinningFramework} Both of our algorithms use the $(1+\varepsilon)$-online thinning framework introduced in \cite{DFG+19}. Fix $\varepsilon \in (0,1)$, they start from an empty set $A$ and, at each step $t$, perform rejection sampling with some target density $\mu_t(\bx)$ on $[0,1)^d$, where $\mu_t(\bx) \in [1 - \varepsilon/2, 1 + \varepsilon/2]$ for all $\bx \in [0,1)^d$.  Specifically, for each step $t=0,1,\ldots,n-1$, they perform the following step:
\begin{enumerate}
    \item Sample $\bx_t\sim\mathsf{Unif}[0,1)^d$ and $\xi_t\sim\mathsf{Unif}[0,1)$. If $\xi_t\leq\mu_t(\bx_t)-\varepsilon/2$, add $\bx_t$ to $A$. Otherwise, sample $\by_t\sim\mathsf{Unif}[0,1)^d$ and add $\by_t$ to $A$.
\end{enumerate}
Note that the strategy never rejects two consecutive samples in this framework. 
Denote by $\bz_t$ the sample added to $A$ at step $t$.
The following is a basic fact about \emph{rejection-sampling}.

\begin{fact}[\cite{DFG+19}]
At each step $t$, the retained point $\bz_t$ has density $\mu_t$.
\end{fact}

For each step $t$, let $A_t$ be the set of $t$ points retained by the strategy in steps $0,\ldots,t-1$, i.e., $A_t = \{\bz_0, \ldots, \bz_{t-1}\}$ and let $A := A_n$ be its output.
For any Haar function $H$, we denote its discrepancy at step $t$ as $\smash{\varphi_t(H) := \sum_{\by\in A_{t}}H(\by)}$. 

An important property of the strategy is that it samples at most $(1 + O(\varepsilon))$ points.
\begin{proposition}\label{prop:NotTooManySamples}
With high probability, any $(1+\varepsilon)$-online thinning strategy as defined above uses at most $(1+O(\varepsilon))n$ samples.
\end{proposition}

\begin{proof}
Note that at each step $t$, the algorithm draws a second sample with probability at most $\varepsilon$, since $\mu_t(\bx)-\varepsilon/2\geq 1-\varepsilon$. Therefore, by Chernoff's inequality (e.g., see \cite{Ver18}) the probability that the algorithm uses more than  $(1 + 2\varepsilon) n$ samples is bounded by $\exp(-\Omega(\varepsilon n))$.
\end{proof}

Finally, we conclude \Cref{sec:preliminaries} by expressing the continuous discrepancy $D_R(A_t)$ in terms of the Haar coefficients of $R$ and the Haar-discrepancy vector $\varphi_t$ for every grid box $R\in\mathcal{R}_{\ell}$.
\begin{observation}\label{obsv:DiscOfRectangle}
    For any grid box $R \in \mathcal{R}_{\ell}$ and any step $t$, we have 
    \[
    D_{R}(A_t) = \sum_{H \in \Pi_{<\ell}}\frac{\langle R, H\rangle}{\lVert H \rVert_{2}^{2}} 
    \cdot \varphi_{t}(H).
    \]
\end{observation}

\begin{proof}
Fix a rectangle $R\in\mathcal{R}_{\ell}$. Since its indicator over $[0,1)^{d}$ (which, by an abuse of notation, we also denote by R) is constant on every dyadic box of level $\ell$, it lies in the span of the Haar functions in $\Pi_{<\ell}$ together with the constant function $H_{\mathbf{0},\mathbf{0}}$. Expanding $R$ in the Haar basis gives
\begin{align*}
D_{R}(A_t) &= \sum_{\bz \in A_{t}}R(\bz) - t|R| \\
&= t \langle R, H_{\mathbf{0},\mathbf{0}} \rangle + \sum_{\bz \in A_{t}}\sum_{H \in \Pi_{<\ell}}\frac{\langle R, H\rangle}{\lVert H \rVert_{2}^{2}}R(\bz) - t|R| \\
&= \sum_{H \in \Pi_{<\ell}}\frac{\langle R, H \rangle}{\lVert H \rVert_{2}^{2}}\sum_{\bz \in A_{t}}R(\bz) = \sum_{H \in \Pi_{<\ell}}\frac{\langle R, H \rangle}{\lVert H \rVert_{2}^{2}}\varphi_{t}(H),
\end{align*}
where the equalities above use that $|A_t| = t$ and $\langle R, H_{\mathbf{0},\mathbf{0}} \rangle = |R|$. 
\end{proof}

\section{QMC Beyond Hardy--Krause via \texorpdfstring{$(1+\varepsilon)n$}{1+eps} Samples}\label{sec:qmc}

In this section, we prove \Cref{thm:main}, restated below.

\MainThm*

The polylogarithmic factor hidden in the $\widetilde{O}_{d,\varepsilon}(\cdot)$ of \Cref{thm:main} can be taken to be $\log^{d} n$, matching the bound of \cite{BJ25a} (see \cite{CJJ26}). In this section, we prove \Cref{thm:main} with a slightly weaker factor of order $\log^{2d} n$, as the algorithm and its analysis are simpler and already contain the main ideas. The $\log^{d} n$ factor can be obtained by combining the algorithm of \Cref{sec:rect-disc} with the random shift introduced here.

\subsection{The Uniformly-Shifted Haar-Thinning Algorithm}\label{subsec:algQMC}

Fix $\varepsilon \in (0,1)$ and let $\ell := \lceil 10 \log_2 n \rceil$. Let $N_\ell := (\ell+1)^d = O_d(\log^d n)$ denote the number of scale vectors $\mathbf{j} \in \{0,1,\ldots,\ell\}^d$. The Haar-thinning strategy of \cite{DFG+19} is an instance of the thinning framework of \Cref{subsec:thinningFramework}, with the target density at step $t$ set to\footnote{\label{footnote:sequences_qmc}The Haar-thinning in \cite{DFG+19} uses a slightly different target density from \eqref{eq:TargetDensity_qmc}, where they set $\ell = \ell(t) = O(\log t)$ to be dependent on the current step $t$. This difference is minor and is mainly to produce a point sequence. Here we use the target density \eqref{eq:TargetDensity_qmc} to keep the algorithm and our analysis cleaner.}
\begin{align}\label{eq:TargetDensity_qmc}
\mu_{t}(\bx) := 1 + \frac{\varepsilon}{2 N_\ell} \sum_{H \in \Pi_{\leq \ell}} \mathsf{sgn}\bigl(-\varphi_{t}(H)\bigr)\,H(\bx),
\end{align}
where $\mathsf{sgn}(\cdot)$ is the sign function with the convention that $\mathsf{sgn}(0) := 0$. 

For every $\bx \in [0,1)^d$ and every scale vector $\mathbf{j}$, exactly one of the Haar functions $H_{\mathbf{j},\mathbf{k}}$ is nonzero at $\bx$, where it takes the value $\pm 1$. Hence at most $N_\ell$ terms of the sum in \eqref{eq:TargetDensity_qmc} are nonzero at $\bx$, and
\begin{align} \label{eq:TargetDensity_bound_qmc}
1 - \varepsilon/2 \leq \mu_t(\bx) \leq 1 + \varepsilon/2 \qquad \text{for every } \bx \in [0,1)^d .
\end{align}
Moreover, $\smash{\int \mu_t = 1}$, as every non-constant Haar function has mean zero. Thus $\mu_t$ is a probability density satisfying the requirement of the $(1+\varepsilon)$-thinning framework in \Cref{subsec:thinningFramework}.

\smallskip
\noindent \textbf{Uniformly-Shifted Haar-Thinning.} Our algorithm is a uniformly-shifted variant of the Haar-thinning strategy above.
It first draws a uniformly random shift $\bs \sim \mathsf{Unif}[0,1)^d$, and then, upon receiving the samples $\bx_1, \bx_2, \ldots$, simulates the Haar-thinning strategy on the shifted sequence $\bx_1 - \bs, \bx_2 - \bs, \ldots$ for $n$ steps, where all shifts in this section are taken modulo $[0,1)^d$: whenever the Haar-thinning strategy retains $\bx_j - \bs$, our algorithm retains the sample $\bx_j$. We write $A_{\bs} \subset [0,1)^d$ for the random set of $n$ samples output by this uniformly-shifted Haar-thinning algorithm.

\smallskip
\noindent \textbf{Shifting the Points Is Equivalent to Shifting the Function.} Fix the shift $\bs$. By construction, the set $A_{\bs} - \bs := \{\bx - \bs : \bx \in A_{\bs}\}$ is the output of the Haar-thinning strategy on the input sequence $\bx_1 - \bs, \bx_2 - \bs, \ldots$. Moreover, for every $f$,
\[
\err(A_{\bs}, f) = \err(A_{\bs} - \bs, f_{\bs}),
\]
where $f_{\bs}(\by) = f(\by + \bs)$ is the shifted function from \Cref{sec:preliminaries}. Finally, as the uniform distribution on $[0,1)^d$ is shift-invariant, the shifted sequence $\bx_1 - \bs, \bx_2 - \bs, \ldots$ is again a sequence of i.i.d.\ uniform samples. Hence, conditioned on $\bs$, the set $A_{\bs} - \bs$ has the same distribution as the output $A$ of the Haar-thinning strategy on the original samples. Taking expectation gives the following.

\begin{observation}[Shifting the function]\label{obs:ShiftingFunction}
For every $f \in L^2([0,1)^d)$,
\begin{align}\label{eq:ShiftingFunction}
\E\bigl[\err(A_{\bs}, f)^2\bigr] = \E_{\bs}\Bigl[\E\bigl[\err(A, f_{\bs})^2\bigr]\Bigr],
\end{align}
where the inner expectation on the right-hand side is over the randomness of the (unshifted) Haar-thinning strategy, for a fixed shift $\bs$.
\end{observation}

\smallskip
\noindent \textbf{Roadmap.}
By \Cref{obs:ShiftingFunction}, it suffices to bound $\E[\err(A, f_{\bs})^2]$ for the unshifted Haar-thinning strategy and a fixed shift $\bs$, and then to average over $\bs$; this is what we do in the remainder of this section. The key new ingredient is the uncorrelation property of low-order Haar discrepancies, which we prove in \Cref{subsec:uncorrelation-Haar} via a symmetry argument. In \Cref{subsec:mainProofWrapper}, we reduce \Cref{thm:main} to bounds on the integration errors of the low- and high-order Haar components of $f_{\bs}$; these bounds are proved in \Cref{subsec:MainTechnicalProofs}, following the streamlined analysis of \cite{CJJ26}. The final passage from the Haar--Besov seminorm of $f_{\bs}$ to $\sigma_{\mathsf{SO}}(f)$ uses the shift-averaged equivalence of \Cref{thm:shift-averaged-mixed-haar-equivalence}.

\subsection{Low-Order Haar Discrepancies Are Uncorrelated}
\label{subsec:uncorrelation-Haar}

The key new ingredient in our analysis is the following uncorrelation property of the discrepancies of distinct low-order Haar functions. 

\begin{lemma}[Uncorrelated Haar discrepancies]\label{lem:UncorrelationLemma}
For every $t \in \{0,1,\ldots,n\}$ and every two distinct Haar functions $H, G \in \Pi_{\leq \ell}$, we have $\mathbb{E}\bigl[\varphi_{t}(H)\,\varphi_{t}(G)\bigr] = 0$.
\end{lemma}

\begin{proof}[Proof of \Cref{lem:UncorrelationLemma}]
The proof is a symmetry argument. We construct an involution $T$ on $[0,1)^d$, i.e., $T\circ T(x) = x$, with two properties. First, the Haar-thinning strategy ``commutes'' with $T$: conditioning on the randomness in the rejection steps, if every sample presented to the strategy is replaced by its image under $T$, then the retained set is replaced by its image under $T$ as well.
Hence the transformation $T$ can be viewed as ``swapping'' pairs of runs of the Haar-thinning strategy with the same probability density,  
so $T$ does not change the distribution of $\varphi_t$. Second, $T$ flips the sign of $\varphi_t(H)$ while leaving $\varphi_t(G)$ unchanged. Together, these two properties show that $\varphi_t(H)\varphi_t(G)$ and $-\varphi_t(H)\varphi_t(G)$ have the same distribution, and hence the same (zero) expectation.

\smallskip
\noindent\textbf{The Map $T$.}
Write $H = H_{\mathbf{j},\mathbf{k}}$ and $G = H_{\mathbf{j}',\mathbf{k}'}$. As $H \neq G$, there is a coordinate $i \in [d]$ in which their one-dimensional factors differ, i.e., $(\mathbf{j}_i,\mathbf{k}_i) \neq (\mathbf{j}'_i,\mathbf{k}'_i)$. Since the roles of $H$ and $G$ are symmetric, we may assume that $\mathbf{j}_i \geq \mathbf{j}'_i$; then $\mathbf{j}_i \geq 1$, as otherwise both factors would equal $h_{0,0}$. To lighten notation, write $j := \mathbf{j}_i$ and $k := \mathbf{k}_i$. Let $I := \mathrm{supp}(h_{j,k}) = [k2^{-(j-1)}, (k+1)2^{-(j-1)})$, and let $I_+$ and $I_-$ denote the left and right halves of $I$, on which $h_{j,k}$ equals $+1$ and $-1$, respectively. Define $T : [0,1)^d \to [0,1)^d$ by $T(\by)_r := \by_r$ for $r \neq i$ and
\[
T(\by)_i :=
\begin{cases}
\by_i + 2^{-j} & \text{if } \by_i \in I_+, \\
\by_i - 2^{-j} & \text{if } \by_i \in I_-, \\
\by_i & \text{otherwise}.
\end{cases}
\]
In other words, $T$ swaps the two halves of $I$ in coordinate $i$ and leaves everything else unchanged. Clearly, $T$ is an involution on $[0,1)^d$, i.e., $T \circ T$ is the identity map. 
We write $T_i : [0,1) \to [0,1)$ for the one-dimensional map that swaps $I_+$ and $I_-$ and fixes $[0,1) \setminus I$, so that $T(\by)_i = T_i(\by_i)$.

\smallskip
\noindent\textbf{Action of $T$ on Low-Order Haar Functions.}
The next claim shows that composing with $T$ permutes the Haar functions in $\Pi_{\leq \ell}$, up to sign.

\begin{claim}\label{claim:HaarUnderT}
There are a map $\chi\colon \Pi_{\leq \ell} \to \{-1,+1\}$ and an involution $\rho \colon \Pi_{\leq \ell} \to \Pi_{\leq \ell}$ such that
\begin{equation}\label{eq:HaarUnderT}
\Gamma \circ T = \chi(\Gamma)\, \rho(\Gamma) \qquad \text{for every } \Gamma \in \Pi_{\leq \ell}.
\end{equation}
Moreover, $\chi(H) = -1$, $\rho(H) = H$, $\chi(G) = +1$, and $\rho(G) = G$.
\end{claim}

\begin{proof}[Proof of \Cref{claim:HaarUnderT}]
Fix $\Gamma \in \Pi_{\leq \ell}$ and let $h_{a,b}$ be its factor in coordinate $i$. Since $T$ only modifies coordinate $i$, the function $\Gamma \circ T$ is obtained from $\Gamma$ by replacing the factor $h_{a,b}$ with $h_{a,b} \circ T_i$. We distinguish three cases.

\smallskip
\noindent\textit{Case 1: $a < j$, or $a \geq j$ and $\mathrm{supp}(h_{a,b}) \cap I = \emptyset$.} If $a < j$, then $h_{a,b}$ is constant on every dyadic interval of length $2^{-a}$ (for $a = 0$, it is constant on all of $[0,1)$), and $I$, being a dyadic interval of length $2^{-(j-1)} \leq 2^{-a}$, is contained in one such interval; hence $h_{a,b}$ is constant on $I$. If instead $\mathrm{supp}(h_{a,b}) \cap I = \emptyset$, then $h_{a,b}$ vanishes on $I$. In either case, since $T_i$ maps $I$ onto itself and fixes its complement, we have $h_{a,b} \circ T_i = h_{a,b}$, and so $\Gamma \circ T = \Gamma$. We set $\chi(\Gamma) := +1$ and $\rho(\Gamma) := \Gamma$.

\smallskip
\noindent\textit{Case 2: $(a,b) = (j,k)$.} Then $h_{a,b} = h_{j,k}$ equals $+1$ on $I_+$, $-1$ on $I_-$, and $0$ outside $I$. As $T_i$ swaps $I_+$ and $I_-$ and fixes the complement of $I$, we have $h_{j,k} \circ T_i = -h_{j,k}$, and so $\Gamma \circ T = -\Gamma$. We set $\chi(\Gamma) := -1$ and $\rho(\Gamma) := \Gamma$.

\smallskip
\noindent\textit{Case 3: $a > j$ and $\mathrm{supp}(h_{a,b}) \cap I \neq \emptyset$.} Since $\mathrm{supp}(h_{a,b})$ and $I$ are intersecting dyadic intervals with $|\mathrm{supp}(h_{a,b})| = 2^{-(a-1)} \leq 2^{-j} = |I|/2$, the support of $h_{a,b}$ is contained in one of the two halves $I_+$ or $I_-$. On this half, $T_i$ is a translation onto the other half, by $2^{-j}$ if the half is $I_+$ and by $-2^{-j}$ if the half is $I_-$. Hence $h_{a,b} \circ T_i$ is the translate of $h_{a,b}$ by $2^{-j}$ or $-2^{-j}$, respectively. Since $2^{-j}$ is an integer multiple of $2^{-(a-1)}$, this translate is again a Haar function at scale $a$, namely $h_{a,b'}$ with $b' = b + 2^{a-1-j}$ or $b' = b - 2^{a-1-j}$, respectively. Thus $\Gamma \circ T = \Gamma'$, where $\Gamma' \in \Pi_{\leq \ell}$ is obtained from $\Gamma$ by replacing the factor $h_{a,b}$ with $h_{a,b'}$. We set $\chi(\Gamma) := +1$ and $\rho(\Gamma) := \Gamma'$.

\smallskip
In all three cases, \eqref{eq:HaarUnderT} holds with $\rho(\Gamma) \in \Pi_{\leq \ell}$, as $\rho$ never changes the scale vector of $\Gamma$. To see that $\rho$ is an involution, apply \eqref{eq:HaarUnderT} twice: $\Gamma = \Gamma \circ T \circ T = \chi(\Gamma)\,\chi(\rho(\Gamma))\,\rho(\rho(\Gamma))$. Since distinct Haar functions are linearly independent, this forces $\rho(\rho(\Gamma)) = \Gamma$.

Finally, the factor of $H$ in coordinate $i$ is $h_{j,k}$, so $H$ falls into Case~2: $\chi(H) = -1$ and $\rho(H) = H$. The factor of $G$ in coordinate $i$ is $h_{\mathbf{j}'_i,\mathbf{k}'_i}$ with $\mathbf{j}'_i \leq j$ and $(\mathbf{j}'_i,\mathbf{k}'_i) \neq (j,k)$. If $\mathbf{j}'_i < j$, then $G$ falls into Case~1. If $\mathbf{j}'_i = j$, then $\mathbf{k}'_i \neq k$, so $\mathrm{supp}(h_{j,\mathbf{k}'_i})$ and $I$ are distinct dyadic intervals of the same length and are therefore disjoint; again $G$ falls into Case~1. Hence $\chi(G) = +1$ and $\rho(G) = G$. 
\end{proof}

\smallskip
\noindent\textbf{The Haar-Thinning Strategy Commutes with $T$.}
Recall from \Cref{subsec:thinningFramework} that at step $t$ the strategy draws $\bx_t \sim \mathsf{Unif}[0,1)^d$ and $\xi_t \sim \mathsf{Unif}[0,1)$, retains $\bx_t$ if $\xi_t \leq \mu_t(\bx_t) - \varepsilon/2$, and otherwise draws and retains a fresh sample $\by_t \sim \mathsf{Unif}[0,1)^d$. 
Consider the \emph{transformed run}, in which the same strategy is executed with the same $\xi_t$ but with $T\bx_t$ and $T\by_t$ in place of $\bx_t$ and $\by_t$, for every $t$. We denote its retained sets, discrepancy vectors, and target densities by $A'_t$, $\varphi'_t$, and $\mu'_t$, respectively.

\begin{claim}\label{claim:Equivariance}
For every $t \in \{0,1,\ldots,n\}$, we have $A'_t = T(A_t)$ and, for every $\Gamma \in \Pi_{\leq \ell}$,
\begin{equation}\label{eq:DiscrepancyUnderT}
\varphi'_t(\Gamma) = \chi(\Gamma)\, \varphi_t(\rho(\Gamma)).
\end{equation}
\end{claim}

\begin{proof}[Proof of \Cref{claim:Equivariance}]
We first note that $A'_t = T(A_t)$ implies \eqref{eq:DiscrepancyUnderT} at step $t$: by \eqref{eq:HaarUnderT}, for every $\Gamma \in \Pi_{\leq \ell}$,
\[
\varphi'_t(\Gamma) = \sum_{\by \in A'_t} \Gamma(\by) = \sum_{\by \in A_t} \Gamma(T\by) = \chi(\Gamma) \sum_{\by \in A_t} \rho(\Gamma)(\by) = \chi(\Gamma)\, \varphi_t(\rho(\Gamma)).
\]
It therefore suffices to show that $A'_t = T(A_t)$ for every $t$, which we do by induction on $t$. For $t = 0$, both sets are empty. Now let $0 \leq t < n$ and suppose that $A'_t = T(A_t)$, so that \eqref{eq:DiscrepancyUnderT} holds at step $t$. Substituting \eqref{eq:DiscrepancyUnderT} and \eqref{eq:HaarUnderT} into the definition \eqref{eq:TargetDensity_qmc} of the target density, we obtain for every $\by \in [0,1)^d$
\begin{align*}
\mu'_t(T\by)
&= 1 + \frac{\varepsilon}{2N_\ell} \sum_{\Gamma \in \Pi_{\leq \ell}} \mathsf{sgn}\bigl(-\varphi'_t(\Gamma)\bigr)\, \Gamma(T\by) \\
&= 1 + \frac{\varepsilon}{2N_\ell} \sum_{\Gamma \in \Pi_{\leq \ell}} \mathsf{sgn}\bigl(-\chi(\Gamma)\,\varphi_t(\rho(\Gamma))\bigr)\, \chi(\Gamma)\, \rho(\Gamma)(\by) \\
&= 1 + \frac{\varepsilon}{2N_\ell} \sum_{\Gamma \in \Pi_{\leq \ell}} \mathsf{sgn}\bigl(-\varphi_t(\rho(\Gamma))\bigr)\, \rho(\Gamma)(\by) \\
&= 1 + \frac{\varepsilon}{2N_\ell} \sum_{\Gamma' \in \Pi_{\leq \ell}} \mathsf{sgn}\bigl(-\varphi_t(\Gamma')\bigr)\, \Gamma'(\by) = \mu_t(\by),
\end{align*}
where the third equality uses that $\mathsf{sgn}$ is odd, so that $\mathsf{sgn}(\chi a)\,\chi = \mathsf{sgn}(a)$ for $\chi \in \{-1,+1\}$, and the fourth re-indexes the sum via the bijection $\Gamma' = \rho(\Gamma)$. Consequently, the transformed run retains $T\bx_t$ at step $t$ if and only if $\xi_t \leq \mu'_t(T\bx_t) - \varepsilon/2 = \mu_t(\bx_t) - \varepsilon/2$, i.e., if and only if the original run retains $\bx_t$; otherwise, both runs retain their second sample, $T\by_t$ and $\by_t$, respectively. In either case, the point retained by the transformed run at step $t$ is the image under $T$ of the point retained by the original run, and hence $A'_{t+1} = T(A_{t+1})$.
\end{proof}

\smallskip
\noindent\textbf{Completing the Proof.}
Since $T$ preserves the Lebesgue measure on $[0,1)^d$, the transformed input stream $(T\bx_t, \xi_t, T\by_t)_{t \geq 0}$ has the same distribution as the original stream $(\bx_t, \xi_t, \by_t)_{t \geq 0}$, namely, all its entries are independent and uniformly distributed. The transformed run is therefore a run of the Haar-thinning strategy on a correctly distributed input stream, and so $\varphi'_t$ has the same distribution as $\varphi_t$ for every $t$. On the other hand, by \Cref{claim:HaarUnderT} and \Cref{claim:Equivariance}, $\varphi'_t(H) = -\varphi_t(H)$ and $\varphi'_t(G) = \varphi_t(G)$. Hence,
\[
\mathbb{E}\bigl[\varphi_t(H)\,\varphi_t(G)\bigr] = \mathbb{E}\bigl[\varphi'_t(H)\,\varphi'_t(G)\bigr] = -\mathbb{E}\bigl[\varphi_t(H)\,\varphi_t(G)\bigr] ,
\]
which implies that $\mathbb{E}[\varphi_t(H)\varphi_t(G)] = 0$, as claimed.
\end{proof}

\begin{remark}\label{rem:SymmetryOfDiscrepancy}
The proof of \Cref{lem:UncorrelationLemma} gives more than uncorrelation. In fact, the same argument shows that $\varphi_t(H)$ and $-\varphi_t(H)$ have the same distribution for every $H \in \Pi_{\leq \ell}$. 
Also, the proof never uses that $H, G \in \Pi_{\leq \ell}$, so \Cref{lem:UncorrelationLemma} holds for any two distinct Haar functions $H$ and $G$.
\end{remark}

We complement \Cref{lem:UncorrelationLemma} with a bound on $\E[\varphi_t(H)^2]$ for any $H \in \Pi_{\leq \ell}$ by $\widetilde{O}_{d,\varepsilon}(1)$, which will be used in the proof of \Cref{lem:LowOrderError}.

\begin{lemma}[Second moments of Haar discrepancies]\label{lem:SecondMomentHaarDisc}
For every step $t$ and every $H \in \Pi_{\leq \ell}$,\footnote{Recall that $\Pi_{\leq \ell}$ excludes the constant function $H_{\mathbf{0}} \equiv 1$, whose discrepancy $\varphi_t(H_{\mathbf{0}}) = |A_t| = t$ is deterministic.}
\[
\mathbb{E}\bigl[\varphi_{t}(H)^{2}\bigr] \leq O(N_\ell^2 / \varepsilon^{2}) = \widetilde{O}_{d,\varepsilon}(1).
\]
\end{lemma}

\begin{proof}[Proof of \Cref{lem:SecondMomentHaarDisc}]
Fix $H \in \Pi_{\leq \ell}$ and set $\lambda := \varepsilon/(6 N_\ell)$. 
For every step $s$, we define 
\[
\Phi_s(H) = \cosh(\lambda \varphi_s(H)). 
\]
As $\varphi_{s+1}(H) = \varphi_s(H) + H(\bz_{s+1})$, using Taylor expansion for $\cosh(\cdot)$ at $\lambda \varphi_s(H)$ gives
\begin{align*}
\Phi_{s+1}(H) 
\leq \Phi_s(H) + \sinh(\lambda \varphi_s(H)) \cdot (\lambda H(\bz_{s+1})) + \Phi_s(H) \cdot (\lambda H(\bz_{s+1}))^2 ,
\end{align*}
where we absorbed the higher-order terms into the second-order term since $|\lambda H(\bz_{s+1})| \leq \lambda \ll 1$.  
Taking the conditional expectation $\E_s[\cdot]$ on both sides above gives 
\begin{align} \label{eq:taylor_bound}
\mathbb{E}_s\big[\Phi_{s+1}(H)\big] \leq \Phi_s(H) + \lambda \sinh(\lambda \varphi_s(H)) \cdot \mathbb{E}_s[H(\bz_{s+1})] + \lambda^2 \Phi_s(H) \cdot \mathbb{E}_{s}\big[ H(\bz_{s+1})^2 \big] .
\end{align}
Note that by \eqref{eq:TargetDensity_qmc} and \eqref{eq:TargetDensity_bound_qmc}, we have 
\begin{align*}
\mathbb{E}_s[H(\bz_{s+1})] & = \langle H, \mu_{s+1} \rangle = -3 \lambda \cdot \mathsf{sgn}(\varphi_s(H)) \|H\|_2^2 \\
\mathbb{E}_{s}\big[ H(\bz_{s+1})^2 \big] & = \langle H^2, \mu_{s+1} \rangle \leq \left(1 + \frac{\varepsilon}{2}\right) \|H\|_2^2 \leq 2 \|H\|_2^2 .
\end{align*}
Plugging these bounds into \eqref{eq:taylor_bound} and using that $|\sinh(x)| \geq \cosh(x) -  1$ for all $x \in \R$ gives 
\begin{align*}
\mathbb{E}_s\big[\Phi_{s+1}(H)\big]  &\leq \Phi_s(H) - 3 \lambda^2 \|H\|_2^2 \cdot \big|\sinh(\lambda \varphi_s(H)) \big|  + 2 \lambda^2 \|H\|_2^2 \cdot \Phi_s(H) \\
& \leq \Phi_s(H) \big(1 - \lambda^2 \|H\|_2^2 \big) + 3 \lambda^2 \|H\|_2^2 .
\end{align*}
Taking expectation on both sides above and using that $\lambda^2 \|H\|_2^2 \ll 1$, an induction argument then gives $\E[\Phi_s(H)] \leq 4$ for every step $s$.

Finally, using $\cosh(x)\geq 1+x^2/2$, we have 
\[
4 \geq \E[\Phi_t(H)] \geq 1 + \frac{\lambda^2}{2} \cdot \E[\varphi_t(H)^2 ] ,
\]
which implies that $\E[\varphi_t(H)^2] \leq O(1/\lambda^2) \leq O(N_\ell^2 / \varepsilon^2)$. 
\end{proof}

\subsection{Proof of \texorpdfstring{\Cref{thm:main}}{Theorem 1.2}}\label{subsec:mainProofWrapper}

Recall from \Cref{obs:ShiftingFunction} that $\E[\err(A_{\bs},f)^2] = \E_{\bs}[\E[\err(A,f_{\bs})^2]]$, where $A = A_n$ is the output of the (unshifted) Haar-thinning strategy and the inner expectation is over its randomness. Throughout, we may assume that $\bar{f} = 0$, and hence $\overline{f_{\bs}} = 0$ for every shift $\bs$,  by replacing $f$ with $f - \bar{f}$. 

As in \cite{CJJ26}, we split the error into the contributions of the low-order Haar components $\Pi_{\leq \ell} f_{\bs}$ and the high-order Haar components $\Pi_{>\ell} f_{\bs}$, and bound the two separately.

\begin{proposition}[Error decomposition]\label{prop:ErrorBound}
For every $g \in L^2([0,1)^d)$ and every finite set $A \subset [0,1)^d$,
\[
\err(A, g)^2 \leq 2\,\err(A,\Pi_{\leq \ell} g)^2 + 2\,\err(A,\Pi_{> \ell}g)^2 .
\]
\end{proposition}

\begin{proof}
Since $\err(A,\cdot)$ is linear and $g = \Pi_{\leq \ell} g + \Pi_{>\ell} g$, we have $\err(A, g) = \err(A,\Pi_{\leq \ell} g) + \err(A,\Pi_{>\ell} g)$, and the claim follows from $(a+b)^2 \leq 2a^2 + 2b^2$.
\end{proof}

The two contributions are bounded in terms of the Haar--Besov seminorm of \Cref{defn:HaarBesovSeminorm}. 
In the following two lemmas, the expectation is only over the randomness of the Haar-thinning strategy. The lemmas are proved in \Cref{subsec:MainTechnicalProofs}.

\begin{lemma}[Low-order error]\label{lem:LowOrderError}
For every $g \in L^2([0,1)^d)$ with $\bar{g} = 0$,
\[
\E\bigl[\err(A,\Pi_{\leq \ell}g)^2\bigr] \leq \frac{\widetilde{O}_{d,\varepsilon}(1)}{n^2}\,\lVert g\rVert_B^2 .
\]
\end{lemma}

\begin{lemma}[High-order error]\label{lem:HighOrderError}
For every $g \in L^2([0,1)^d)$,
\[
\E\bigl[\err(A,\Pi_{> \ell}g)^2\bigr] \leq \frac{2^{1-\ell}}{n}\,\lVert g\rVert_B^2 \leq \frac{1}{\poly(n)}\,\lVert g\rVert_B^2 .
\]
\end{lemma}

We prove these lemmas in \Cref{subsec:MainTechnicalProofs}. Here, we first prove \Cref{thm:main} assuming the lemmas.

\begin{proof}[Proof of \Cref{thm:main} assuming the lemmas]
Fix a shift $\bs$ and apply \Cref{prop:ErrorBound} to $g = f_{\bs}$, which has mean zero. By \Cref{lem:LowOrderError} and \Cref{lem:HighOrderError},
\[
\E\bigl[\err(A,f_{\bs})^2\bigr]
\leq 2\,\E\bigl[\err(A,\Pi_{\leq \ell}f_{\bs})^2\bigr] + 2\,\E\bigl[\err(A,\Pi_{> \ell}f_{\bs})^2\bigr]
\leq \frac{\widetilde{O}_{d,\varepsilon}(1)}{n^2}\,\lVert f_{\bs}\rVert_B^2 ,
\]
where the expectations are over the randomness of the strategy. Taking the expectation over $\bs \sim \mathsf{Unif}[0,1)^d$ and using \Cref{obs:ShiftingFunction} and \Cref{thm:shift-averaged-mixed-haar-equivalence}, we obtain
\[
\E\bigl[\err(A_{\bs},f)^2\bigr] = \E_{\bs}\Bigl[\E\bigl[\err(A,f_{\bs})^2\bigr]\Bigr] \leq \widetilde{O}_{d,\varepsilon}\bigl(\sigma_{\mathsf{SO}}(f)^2/n^2\bigr),
\]
which gives the error bound of \Cref{thm:main}.

It remains to verify the other properties claimed in \Cref{thm:main}. The uniformly-shifted Haar-thinning algorithm is online and never discards two consecutive samples, by the definition of the thinning framework in \Cref{subsec:thinningFramework}.  \Cref{prop:NotTooManySamples} shows that it retains $n$ points from at most $(1+\varepsilon)n$ samples with high probability. Finally, each step evaluates $\mu_t$ at a single point $\bx$. Since at most $N_\ell$ Haar functions in $\Pi_{\leq \ell}$ are nonzero at $\bx$ (one for each scale vector), $\mu_t(\bx)$ is a sum of at most $N_\ell$ terms, and retaining the point $\bx$ changes $\varphi_t(H)$ only for these Haar functions $H$. Storing the nonzero entries of $\varphi_t$ in a dictionary, each step therefore takes $O(N_\ell) = O_d(\log^d n)$ time.
\end{proof}

\subsection{Bounds on the Low- and High-Order Components}\label{subsec:MainTechnicalProofs}

We now prove \Cref{lem:LowOrderError} and \Cref{lem:HighOrderError}  respectively in the next two subsections.

\subsubsection{Bounds on the Low-Order Component}
\label{subsubsec:BoundOnLowFrequencyComponent}

Assume $\bar{g} = 0$ as in \Cref{lem:LowOrderError}, which implies that $\langle g, H_{\mathbf{0}}\rangle= 0$ for the constant Haar function $H_{\mathbf{0}}$. Writing $\Pi_{\leq \ell}g = \sum_{H \in \Pi_{\leq \ell}} c_H H$ with $c_H := \langle g, H\rangle/\lVert H\rVert_2^2$, we have
\begin{equation}\label{eq:LowOrderErrorAsInnerProduct}
\err(A,\Pi_{\leq \ell}g)
= \frac{1}{n}\sum_{\by \in A} (\Pi_{\leq \ell}g)(\by)
= \frac{1}{n}\sum_{H \in \Pi_{\leq \ell}} c_H \sum_{\by \in A} H(\by)
= \frac{1}{n}\,\langle c, \varphi_n\rangle ,
\end{equation}
where $c = (c_H)_{H \in \Pi_{\leq \ell}}$. Thus the second moment of the low-order error is a quadratic form of the second-moment matrix $\E[\varphi_n\varphi_n^{\top}]$, which can be controlled by \Cref{lem:UncorrelationLemma} and \Cref{lem:SecondMomentHaarDisc}.

\begin{proof}[Proof of \Cref{lem:LowOrderError}]
By \eqref{eq:LowOrderErrorAsInnerProduct}, \Cref{lem:UncorrelationLemma}, and \Cref{lem:SecondMomentHaarDisc}, we can bound 
\begin{align*}
\E\bigl[\err(A,\Pi_{\leq \ell}g)^2\bigr]
&= \frac{1}{n^2}\,\E\bigl[\langle c, \varphi_n\rangle^2\bigr]
\leq \frac{1}{n^2} \, \|c\|_2^2 \, \max_{H \in \Pi_{\leq \ell}} \E\bigl[\varphi_n(H)^2\bigr] \\
&\leq \frac{\widetilde{O}_{d,\varepsilon}(1)}{n^2} \sum_{H \in \Pi_{\leq \ell}} \frac{|\langle g, H\rangle|^2}{\lVert H\rVert_2^4}\\
&\leq \frac{\widetilde{O}_{d,\varepsilon}(1)}{n^2}\, \|g\|_B^2 . 
\end{align*}
This completes the proof of the lemma. 
\end{proof}

\subsubsection{Bounds on the High-Order Component}
\label{subsubsec:BoundOnHighFrequencyComponent}

Throughout this subsection, we use $\E_{t-1}[\cdot]$ to denote the expectation conditioned on the randomness of the Haar-thinning strategy in steps $\leq t-1$.

Our main observation is that the additive terms in \eqref{eq:TargetDensity_qmc}, and hence $\mu_t - \bf{1}$, are all orthogonal to Haar functions of order $> \ell$. This gives the following lemma.   

\begin{lemma}[Orthogonality of increments]\label{lem:OrthogonalityOfIncrementHighOrder}
Set $h := \Pi_{>\ell}g$. For every step $t$,
\begin{align*}
\mathbb{E}_{t-1}\bigl[h(\bz_t) \bigr] = 0
\qquad \text{and} \qquad
\mathbb{E}_{t-1}\bigl[h(\bz_t)^2 \bigr] \leq (1 + \varepsilon / 2)\, \|h\|_2^2 .
\end{align*}
Consequently,
\[
\E\bigl[\err(A,h)^2\bigr]
=
\frac{1}{n^2}\sum_{t=0}^{n-1}\E\bigl[h(\bz_t)^2\bigr]
\leq
\frac{1+\varepsilon/2}{n}\,\lVert h\rVert_2^2 .
\]
\end{lemma}

\begin{proof}
Since $\mu_t - \bf{1}$ is orthogonal to Haar functions of order $> \ell$, we have $\langle h, \mu_t - \bf{1} \rangle = 0$, and thus
\[
\mathbb{E}_{t-1}[h(\bz_t)] = \langle h, \mu_t \rangle = \langle h, \bf{1} \rangle = 0 ,
\]
where we used that the high-order Haar component $h$ has mean zero. 
For the second bound, using that $\mu_t(\bx) \leq 1 +  \varepsilon/2$ for every $\bx \in [0,1)^d$, we have
\[
\mathbb{E}_{t-1}\bigl[h(\bz_t)^2 \bigr] = \langle h^2, \mu_t \rangle \leq (1 +  \varepsilon/2) \, \|h\|_2^2 .
\]
For the consequence, expand
\[
\E\left[\left(\sum_{t=0}^{n-1} h(\bz_t)\right)^2\right]
= \sum_{t=0}^{n-1} \E\bigl[h(\bz_t)^2\bigr] + 2\sum_{0 \leq t < t' \leq n-1} \E\bigl[h(\bz_t)\,h(\bz_{t'})\bigr] .
\]
For $t < t'$, the value $h(\bz_t)$ is already fixed before step $t'$. Therefore,
\[
\E\bigl[h(\bz_t)h(\bz_{t'})\bigr]
=
\E\left[h(\bz_t)\,\mathbb{E}_{t'-1}[h(\bz_{t'})]\right]
=0.
\]
Using the bound on the conditional second moments now gives the claim.
\end{proof}

\begin{proof}[Proof of \Cref{lem:HighOrderError}]
By \Cref{lem:OrthogonalityOfIncrementHighOrder}, 
\[
\E[\err(A, \Pi_{> \ell} g)^2] \leq \left(1 + \frac{\eps}{2}\right) \frac{\|\Pi_{> \ell} g\|_2^2}{n}.
\]
It therefore remains only to bound the $L_2$-norm of the high-order component $\Pi_{> \ell} g$.

Let $H = H_{\mathbf{j}, \mathbf{k}}$ be a Haar function appearing in $\Pi_{> \ell} g$, so that $H \notin \Pi_{\leq \ell}$. Then there is some coordinate $i$ for which $\mathbf{j}_i \geq \ell + 1$. Since $H$ is a tensor product of one-dimensional Haar functions and every one-dimensional factor has $L_2$-norm at most $1$, we have
\[
\|H\|_2^2 = \prod_{r=1}^d \|h_{\mathbf{j}_r, \mathbf{k}_r}\|_2^2 \leq \|h_{\mathbf{j}_i, \mathbf{k}_i}\|_2^2 = \frac{1}{2^{\mathbf{j}_i - 1}} \leq \frac{1}{2^\ell}.
\]
Using the orthogonality of the Haar basis, we may therefore write
\begin{align*}
\lVert \Pi_{>\ell} g\rVert_2^2
&=
\sum_{H\notin \Pi_{\leq \ell}}
\frac{\lvert \langle g,H\rangle\rvert^2}{\lVert H\rVert_2^2} 
=
\sum_{H\notin \Pi_{\leq \ell}}
\lVert H\rVert_2^2 \cdot
\frac{\lvert \langle g,H\rangle\rvert^2}{\lVert H\rVert_2^4} 
\leq
2^{-\ell}
\sum_{H\notin \Pi_{\leq \ell}}
\frac{\lvert \langle g,H\rangle\rvert^2}{\lVert H\rVert_2^4}\\ &\leq
2^{-\ell}\lVert g\rVert_B^2,
\end{align*}
where the last inequality follows from the definition of $\|g\|_B$. Combining the preceding results and using that $\eps < 1$ and $\ell = \lceil 10\log_2 n \rceil$, we obtain
\[
\E[\err(A, \Pi_{> \ell} g)^2] \leq \left(1 + \frac{\eps}{2}\right) \frac{2^{-\ell}\lVert g\rVert_B^2}{n}
\leq \frac{2^{1-\ell}\lVert g\rVert_B^2}{n} \leq \frac{\|g\|_B^2}{\poly(n)} . \qedhere
\]
\end{proof}

\section{A \texorpdfstring{$(1+\varepsilon)$}{1+eps}-Thinning Strategy for \texorpdfstring{$O_d(\log^{d+1}n)$}{O(logd+1n)} Star Discrepancy}
\label{sec:rect-disc}

In this section, we prove \Cref{thm:GridRectangleThinning}. For completeness we include its statement.

\gridRectangleThinning*

Although $D^{\star}(A)$ is the maximum continuous discrepancy over all anchored boxes, it suffices to only control those of the grid boxes $\mathcal{R}_{\ell}$ for $\ell=O_d(\log(nd))$ (recall \Cref{sec:preliminaries}), since they approximate every anchored box from inside and outside within a volume difference of at most $1/\mathsf{poly}(n, d)$. 

\smallskip
\noindent \textbf{Road Map.} We give our algorithm, a variant of Haar-thinning, in \Cref{subsec:RectDiscAlg}. 
Key to our analysis is a concentration bound for linear functions of Haar discrepancies, which we prove in \Cref{subsec:conc-linear-func-Haar}. Finally, we use this bound to analyze our algorithm and prove \Cref{thm:GridRectangleThinning} in \Cref{subsec:RectDiscAnalysis}.

\subsection{The Linear-Feedback Haar-Thinning Algorithm}\label{subsec:RectDiscAlg}

Fix $\varepsilon \in (0,1)$ and set $\ell := O(\log (nd))$.\footnote{For a cleaner presentation, the algorithm we describe here produces a point set instead of a point sequence. To produce a point sequence as stated in \Cref{thm:GridRectangleThinning}, one can slightly adjust the algorithm as in \Cref{footnote:sequences_qmc}.} 
Our algorithm for \Cref{thm:GridRectangleThinning} is a variant of the Haar-thinning strategy from \Cref{subsec:algQMC} whose target density at step $t$ is set to be 
\begin{equation} \label{eq:linear-feedback-Haar}
    \mu_t(\bx) = 1 - \frac{\varepsilon}{2B}\Phi_{t}(\bx), \quad \text{ where} \quad   \Phi_t(\bx) = \sum_{H \in \Pi_{\leq \ell}}\varphi_t(H)H(\bx),
\end{equation}
where we set $B =O_{d,\varepsilon}(\log^{d+1}n)$ for a sufficiently large constant.

Note that the main difference from the Haar-thinning strategy \eqref{eq:TargetDensity_qmc} is the replacement of $\mathsf{sgn}(\varphi_t(H))$ by $\varphi_t(H)$. 
The following proposition shows that \eqref{eq:linear-feedback-Haar} gives a negative drift that is proportional to the Haar discrepancies (and hence the name {\em linear-feedback} Haar-thinning).

\begin{proposition}[Proportional negative drift]\label{prop:HaarMeanAtStept}
At every step $t$, for every $H\in \Pi_{\leq \ell}$,
\[
\mathbb{E}_{t-1}[H(\bz_t)]
=
-\frac{\varepsilon}{2} \, \varphi_{t-1}(H)\lVert H\rVert_2^2,
\]
where $\mathbb{E}_{t-1}[\cdot]$ denotes expectation conditioned on all points sampled till step $t-1$.
\end{proposition}

\begin{proof}
As each Haar function is mean-zero, we can write 
\begin{align*}
    \mathbb{E}_{t-1}[H(\bz_t)]
    &= \int_{[0,1)^d}H(\bx)\bigl(1 -\frac{\varepsilon}{2B}\Phi_{t-1}(\bx)\bigr)\,\mathrm d\bx\\
    &= -\frac{\varepsilon}{2B}
    \sum_{G\in \Pi_{\leq \ell}}
    \varphi_{t-1}(G)
    \langle H,G\rangle = - \frac{\varepsilon}{2B}\varphi_{t-1}(H)\lVert H \rVert_{2}^{2},
\end{align*}
where the last equality follows from the orthogonality of Haar functions. 
\end{proof}

\smallskip
\noindent \textbf{$(1+\varepsilon)$-Thinning Condition.} Note that for $\mu_t(\bx)$ in \eqref{eq:linear-feedback-Haar} to be a valid $(1+\varepsilon)$-thinning strategy, i.e., $\mu_t(\bx) \in [1 - \varepsilon/2, 1+\varepsilon/2]$ (recall \Cref{subsec:thinningFramework}), we need the condition that $|\Phi_t(\bx)| \leq B$. 

\begin{abort}
\label{rule:StringentSaturation}
We say the algorithm $\mathsf{FAILs}$ at time $t$ if $|\Phi_t(\bx)| > B$ for any $\bx \in [0,1)^{d}$. 
\end{abort}

We will show later in \Cref{lem:RectDiscAlgDoesNOTAbort} that the algorithm never \textsf{FAILs} with high probability.

For the purpose of the analysis, we will analyze the following modified algorithm: whenever the above algorithm \textsf{FAILs} at a time $t$, we will freeze the process afterwards, i.e., stop retaining/rejecting any new samples, and set $\varphi_{t'}(H) = \varphi_t(H)$ for every $t' > t$ and every $H \in \Pi_{\leq \ell}$. 
Clearly, the modified algorithm only differs from the original algorithm when it \textsf{FAILs}. 

\subsection{Concentration for Linear Functions of Haar Discrepancies}
\label{subsec:conc-linear-func-Haar}

Key to our analysis is the following concentration bound for linear functions of Haar discrepancies, which we will crucially use in \Cref{subsec:RectDiscAnalysis} to prove \Cref{thm:GridRectangleThinning}. 

\begin{restatable}{lemma}{freedmanTailBound}
\label{lem:TailboundViaFreedman}
For every step $t$, vector $a\in\smash{\mathbb{R}^{\Pi_{\leq \ell}}}$, and $\xi\geq 0$, the modified algorithm satisfies
\[
\Pr\bigl(\langle a,\varphi_t\rangle\geq\xi\bigr)
\leq
\exp\Bigl(
\frac{-\Omega(\xi^2)}
{(B/\varepsilon)\lVert a\rVert_2^2+\xi\lVert a\rVert_1}
\Bigr).
\]
\end{restatable}

We prove \Cref{lem:TailboundViaFreedman} in the remainder of this subsection using Freedman's inequality.

\begin{theorem}[Freedman~\cite{Fre75}]
\label{fact:Freedman}
Let $Y_0,\ldots,Y_n$ be a martingale with respect to
$X_1,\ldots,X_n$ such that $
|Y_k-Y_{k-1}|\leq M $
for every $k$. Let
\[
W_k
=
\sum_{j=1}^k
\mathbb{E}_{j-1}\bigl[(Y_j-Y_{j-1})^2\bigr]
=
\sum_{j=1}^k
\operatorname{Var}\bigl[Y_j\mid X_1,\ldots,X_{j-1}\bigr],
\]
where $\mathbb{E}_{j-1}[\cdot]$ denotes conditioning on step $j-1$. Then, for every $\lambda\geq 0$ and $\sigma^2\geq 0$,
\[
\Pr\left(
    |Y_n-Y_0|\geq\lambda
    \ \text{and}\
    W_n\leq\sigma^2
\right)
\leq
2\exp\left(
    -\frac{\lambda^2}
    {2\left(\sigma^2+M\lambda/3\right)}
\right).
\]
\end{theorem}

The proof consists of two steps: \emph{(i)} we view $\langle a,\varphi_t\rangle$ as a martingale and express it as a sum of mean-zero increments $Y_{s,t}$ for $s < t$, \emph{(ii)} we bound its quadratic variation and maximum possible increment in each step, and apply Freedman's inequality.

\smallskip
\noindent \textbf{The Random Process.} 
Define $Y_t := \langle a, \varphi_t \rangle$,  
which can be naturally expressed as
\begin{equation}\label{eq:RandomProcessNaturalIncrement}
    Y_t
    =
    \sum_{s<t}(Y_{s+1}-Y_s)
    =
    \sum_{s<t}\langle a,\varphi_{s+1}-\varphi_s\rangle
    =:
    \sum_{s<t}\langle a,\Delta\varphi_s\rangle.
\end{equation}
However, $\Delta\varphi_s(H) = H(\bz_{s})$, which has nonzero conditional mean by \Cref{prop:HaarMeanAtStept}. To bypass this issue, we \emph{center} the increments to express $Y_t$ as a sum of mean-zero terms $Y_{t,s}$ for $s<t$.

\begin{proposition}\label{prop:CenteringRandomProcess}
The random process $Y_{t}$ for $t \geq 1$ can be expressed as
\[
Y_t=\sum_{s<t}Y_{t,s}
\quad\text{where}\quad
Y_{t,s}
:=
\sum_{H\in \Pi_{\leq \ell}}
\widebar{\Delta\varphi}_s(H)\cdot a_H
\Bigl(1-\frac{\varepsilon\lVert H\rVert_2^2}{2B}\Bigr)^{t-s-1},
\]
where $\widebar{\Delta\varphi}_s := \Delta\varphi_s - \E_{s-1}[\Delta\varphi_s]$ denotes its centering.
\end{proposition}

\begin{proof}
The change of $Y$ at step $t - 1$ can be expressed as:
\begin{align*}
Y_{t} - Y_{t-1} 
&= \langle a, \Delta \varphi_{t-1} - \mathbb{E}_{t-1}[\Delta \varphi_{t-1}]\rangle + \mathbb{E}_{t-1}\bigl[\langle a, \Delta \varphi_{t-1}\rangle\bigr] \\
&= \langle a, \widebar{\Delta\varphi}_{t-1}\rangle + \sum_{H \in \Pi_{\leq \ell}}\mathbb{E}_{t-1}[H(\bz_t)]\cdot a_{H}\\
&= \langle a, \widebar{\Delta\varphi}_{t-1}\rangle - \frac{\varepsilon}{2B}\sum_{H \in \Pi_{\leq \ell}}\varphi_{t-1}(H)\lVert H \rVert_{2}^{2}\cdot a_{H},
\end{align*}
where the first equality follows by \eqref{eq:RandomProcessNaturalIncrement} and the final equality follows by \Cref{prop:HaarMeanAtStept}. Re-arranging $Y_{t-1}$ to the other side, we get
\[
Y_{t} = \langle a, \widebar{\Delta\varphi}_{t-1}\rangle + \sum_{H \in \Pi_{\leq \ell}}\Bigl(1 - \frac{\varepsilon \lVert H \rVert_{2}^{2}}{2B}\Bigr)\varphi_{t-1}(H)\cdot a_{H}.
\]
Recursing on $\smash{\varphi_{t-1}}$ over all previous steps completes the proof of \Cref{prop:CenteringRandomProcess}. 
\end{proof}

We next bound the quadratic variation of $Y_t$ and the maximum size of $|Y_{t,s}|$ for all $s<t$. 
It suffices to consider the steps before the algorithm \textsf{FAILs}, as all later increments of $Y_t$ are zero. 

\begin{proposition}[Quadratic variation]\label{prop:SumOfVariance}
The quadratic variation of $Y_t$ satisfies
\[
\sum_{s<t}\mathbb{E}_s\bigl[Y_{t,s}^2\bigr]
\leq
O\Bigl(\frac{B}{\varepsilon} \, \lVert a\rVert_2^2\Bigr).
\]
\end{proposition}

\begin{proof}
For any step $s$ before the algorithm \textsf{FAILs}, we have
\begin{align*}
    \mathbb{E}_{s}[Y_{t,s}^{2}] 
    & = \mathbb{E}_s \Big[ \Big( \sum_{H\in \Pi_{\leq \ell}}
\widebar{\Delta\varphi}_s(H)\cdot a_H
\Bigl(1-\frac{\varepsilon\lVert H\rVert_2^2}{2B}\Bigr)^{t-s-1} \Big)^2 \Big] \\
& = \mathbb{E}_s \Big[ \Big( \sum_{H\in \Pi_{\leq \ell}}
\big( H(\bz_{s+1}) - \mathbb{E}_s\big[H(\bz_{s+1}) \big] \big) \cdot a_H
\Bigl(1-\frac{\varepsilon\lVert H\rVert_2^2}{2B}\Bigr)^{t-s-1} \Big)^2 \Big] \\
& \leq \mathbb{E}_s \Big[ \Big( \sum_{H\in \Pi_{\leq \ell}}
H(\bz_{s+1}) \cdot a_H
\Bigl(1-\frac{\varepsilon\lVert H\rVert_2^2}{2B}\Bigr)^{t-s-1} \Big)^2 \Big]  ,
\end{align*}
where the last inequality uses the fact that the variance is at most the second moment. 

As we have assumed that the algorithm doesn't \textsf{Fail} at step $s$, we have $\mu_{s+1}(\bx)\leq 1+\varepsilon/2$, and hence the expectation of any nonnegative function under the density $\mu_{s+1}$ is at most $(1+\varepsilon/2)$ times its expectation under the uniform distribution. Hence for $\bz \sim \Unif[0,1)^d$, we have
\begin{align*}
    \mathbb{E}_{s}[Y_{t,s}^{2}] &\leq  (1 + \varepsilon/2) \cdot\mathbb{E}_s \Big[ \Big( \sum_{H\in \Pi_{\leq \ell}}
H(\bz) \cdot a_H
\Bigl(1-\frac{\varepsilon\lVert H\rVert_2^2}{2B}\Bigr)^{t-s-1} \Big)^2 \Big] \notag \\
    &= 
    (1 + \varepsilon/2) \cdot \sum_{H \in \Pi_{\leq \ell}}
    a_{H}^{2}\lVert H \rVert_{2}^{2}\Bigl(1 - \frac{\varepsilon \lVert H \rVert_{2}^{2}}{2B}\Bigr)^{2(t-s-1)} ,
\end{align*}
where the equality uses the orthogonality of Haar functions and $\E[H(\bz)^2] = \|H\|_2^2$.
Then, summing over $s<t$ and bounding the resulting geometric series gives 
\vspace{-0.03em}
\begin{align*}
\sum_{s < t}\mathbb{E}_{s}[Y_{t,s}^{2}] &\leq 
    2 \sum_{s < t}\sum_{H \in \Pi_{\leq \ell}}
    a_{H}^{2}\lVert H \rVert_{2}^{2}\Bigl(1 - \frac{\varepsilon \lVert H \rVert_{2}^{2}}{2B}\Bigr)^{2(t-s-1)}  \\
    & = 
    2 \sum_{H \in \Pi_{\leq \ell}}
    a_{H}^{2}\lVert H \rVert_{2}^{2}\sum_{s < t}\Bigl(1 - \frac{\varepsilon \lVert H \rVert_{2}^{2}}{2B}\Bigr)^{2(t-s-1)} \leq O\Bigl(\frac{B}{\varepsilon} \, \lVert a \rVert_{2}^{2}\Bigr).  \qedhere
\end{align*}
\end{proof}
We next bound the maximum possible size of $|Y_{t,s}|$ for any $s<t$.

\begin{proposition}\label{prop:AbsoluteValueofIncrement}
    At every step $s < t$, the absolute value of the increment $Y_{t,s}$ satisfies
    \[
    |Y_{t,s}| \leq 2\lVert a \rVert_{1}.
    \]
\end{proposition}

\begin{proof}
Applying the triangle inequality to $Y_{t,s}$ gives
\begin{equation}\label{eq:FirstEquationAbsoluteValue}
|Y_{t,s}| = \Bigl|\sum_{H\in \Pi_{\leq \ell}}
\widebar{\Delta\varphi}_s(H)\cdot a_H
\Bigl(1-\frac{\varepsilon\lVert H\rVert_2^2}{2B}\Bigr)^{t-s-1}\Bigr| \leq \sum_{H}|\widebar{\Delta \varphi_{s}}(H)|\cdot |a_{H}| .
\end{equation}
For any Haar function $H \in \Pi_{\leq \ell}$, as $|\Delta \varphi_s(H)| = |H(\bz_{s+1})| \leq 1$, its centered version satisfies \begin{align*}
\big|\widebar{\Delta \varphi_{s}}(H) \big| \leq \big|\Delta \varphi_{s}(H) \big| + \big|\E_s \big[\Delta \varphi_{s}(H) \big]\big| \leq 2.
\end{align*}
The bound in the proposition then follows immediately. 
\end{proof}

Now we are ready to prove \Cref{lem:TailboundViaFreedman}.

\begin{proof}[Proof of \Cref{lem:TailboundViaFreedman}]
    Fix $a\in\smash{\mathbb{R}^{\Pi_{\leq \ell}}}$, $\xi\geq 0$, and a step $t\geq 1$. Applying Freedman's inequality (see \Cref{fact:Freedman}) to bound the probability of $Y_{t}$ deviating from $Y_{0}$ by $\xi$ 
    when the quadratic variation is at most $\sigma^{2} = O((B/\varepsilon))\lVert a\rVert_2^2)$ (which always holds due to \Cref{prop:SumOfVariance})
    gives
    \begin{align*}
    \Pr\bigl(Y_{t} - Y_{0} \geq \xi\bigr) &= \Pr\biggl(Y_{t} - Y_{0} \geq \xi \text{ and } \sum_{s<t}\mathbb{E}_s\bigl[Y_{t,s}^2\bigr]
\leq
\sigma^{2}\biggr)
\leq \exp\Bigl(\frac{-\Omega(\xi^{2})}{(B/\varepsilon)\lVert a \rVert_{2}^{2} + \xi \lVert a \rVert_{1}}\Bigr) .  \qedhere
\end{align*}
\end{proof}

\subsection{Wrapping Things Up}\label{subsec:RectDiscAnalysis}

We now use \Cref{lem:TailboundViaFreedman} to complete the analysis of the modified algorithm and prove \Cref{thm:GridRectangleThinning}. 
We start by showing that each grid box $R\in\mathcal{R}_{\ell}$ has small $D_R(A_t)$ at every step $t$. 

\begin{lemma}[Discrepancy of grid boxes]\label{lem:lowFixRectDisc}
With probability at least $1-1/\mathsf{poly}(n^d)$, for every step $t$ and every grid box $R\in\mathcal{R}_{\ell}$, its continuous discrepancy under $A_t$ satisfies
\[
D_{R}(A_t)\leq O_{d,\varepsilon}(\log^{d+1}n).
\]
\end{lemma}

\begin{proof}
Fix any grid box $R\in\mathcal{R}_{\ell}$ and step $t$. Recall from \Cref{obsv:DiscOfRectangle} that 
\begin{align}\label{eq:ContinuousDiscofR}
    D_{R}(A_t) &= \sum_{H \in \Pi_{\leq \ell}}\frac{\langle R, H\rangle}{\lVert H \rVert_{2}^{2}} \, \varphi_{t}(H) \notag.
\end{align}
To control $D_{R}(A_t)$ using \Cref{lem:TailboundViaFreedman}, we consider the vector $a \in \R^{\Pi_{\leq \ell}}$ with $a_H := \langle R, H \rangle / \|H\|_2^2$. 

Note that for each $i \in [d]$ and each level $\mathbf{j}_i$, at most two Haar functions can have nonzero coefficients: these are the Haar functions whose supports contain one of the two endpoints of the interval defining $R$ in coordinate $i$.
As there are $\ell$ levels in each coordinate, at most $\smash{(2\ell)^d=O_d(\log^d(nd))}$ of the $a_H$ can be non-zero, each with $|a_H| \leq 1$. This shows that both $\smash{\lVert a\rVert_1}$ and $\smash{\lVert a\rVert_2^2}$ are at most $\smash{O_d(\log^d(nd))}$. Plugging these bounds into \Cref{lem:TailboundViaFreedman} for $\xi \geq 0$, we get
\begin{align*}
\Pr(D_{R}(A_t) \geq \xi) &= \Pr(\langle a, \varphi_{t}\rangle \geq \xi)\\
&\leq \exp\Bigl(\frac{\Omega(-\xi^{2})}{(B/\varepsilon)\lVert a \rVert_{2}^{2} + \lVert a \rVert_{1}}\Bigr) \leq \exp\Bigl(\frac{\Omega_{d,\varepsilon}(-\xi^{2})}{\log^{2d+1}n}\Bigr),
\end{align*}
where the final inequality follows by substituting $B=O_{d,\varepsilon}(\log^{d+1}n)$ and using the bounds on $\lVert a\rVert_1$ and $\lVert a\rVert_2^2$. Setting $\xi = O_{d,\varepsilon}(\log^{d+1}n)$ for a sufficiently large constant and taking the union bound over all of the $2^{\ell d}$ grid boxes in $\mathcal{R}_\ell$ and all steps $t$ completes the proof.
\end{proof}

Next we show that with high probability, the algorithm never \textsf{FAILs}. 

\begin{lemma}\label{lem:RectDiscAlgDoesNOTAbort}
With probability at least $1 - 1/\mathsf{poly}(n^{d})$, the algorithm never \textsf{FAILs}.
\end{lemma}
\begin{proof}
Consider all of the $2^{\ell d} = (nd)^{O(d)}$ dyadic boxes $I$ of edge length $2^{-\ell}$ in every dimension. Note that each Haar function $H \in \Pi_{\leq \ell}$ is constant on any such $I$, and hence to check the condition $\smash{|\Phi_{t}(\bx)|<B}$ for all $\bx\in[0,1)^d$, it suffices to only check it for one fixed point $\bx_I \in I$ for each $I$. 

Fix a dyadic box $\smash{I\in\mathcal{D}_{\leq \ell}^d}$ and some $\bx_I \in I$. We will express $\smash{\Phi_t(\bx_I)}$ as a linear function of $\varphi_t$ and then apply \Cref{lem:TailboundViaFreedman} to bound the probability that it exceeds $B$. Since
\[
\Phi_{t}(\bx_I) = \sum_{H \in \Pi_{\leq \ell}}\varphi_{t}(H)H(\bx_I) ,
\]
we define $a\in\smash{\mathbb{R}^{\Pi_{\leq \ell}}}$ by $a_H=H(\bx_I)$. In each coordinate, $\bx_I$ lies in the support of at most $\ell$ Haar functions, so $a$ has at most $\ell^d$ nonzero entries. Since each such entry is in $\smash{\{\pm1\}}$, both $\smash{\lVert a\rVert_1}$ and $\smash{\lVert a\rVert_2^2}$ are at most $\smash{\ell^d=O_d(\log^d(nd))}$. Plugging these bounds into \Cref{lem:TailboundViaFreedman} for $\xi \geq 0$, we get
\begin{align*}
\Pr\bigl(\Phi_{t}(\bx) \geq \xi\bigr) &= \Pr\bigl(\langle a, \varphi_{t}\rangle \geq \xi\bigr) \\
&\leq \exp\Bigl(\frac{\Omega(-\xi^{2})}{(B/\varepsilon)\lVert a \rVert_{2}^{2} + \lVert a \rVert_{1}}\Bigr) \leq \exp\Bigl(\frac{\Omega_{d,\varepsilon}(-\xi^{2})}{\log^{2d+1}n}\Bigr).
\end{align*}
Setting $\xi = B = O_{d,\varepsilon}(\log^{d+1}n)$ for a sufficiently large constant and then take a union bound over all the dyadic boxes $I$ completes the proof.
\end{proof}

We are now ready to prove \Cref{thm:GridRectangleThinning}.

\begin{proof}[Proof of \Cref{thm:GridRectangleThinning}]
By \Cref{thm:GridRectangleThinning} and \ref{lem:lowFixRectDisc}, with high probability, the algorithm never \textsf{FAILs} and that every grid box $R \in \mathcal{R}_\ell$ has continous discrepancy at most $D_R(A_t) \leq O_{d,\varepsilon}(\log^{d+1} n)$ at each step $t$. 
Since any anchored box can be approximated from inside and outside by grid boxes within a volume difference of at most $1/\poly(n^d)$, this proves the theorem. 
\end{proof}

\section*{Acknowledgements}
We thank Jiaheng Chen for helpful discussions.

\bibliography{bib}
\bibliographystyle{alphaurl}

\appendix

\end{document}